\documentclass{article}

\usepackage{arxiv}

\usepackage[utf8]{inputenc}
\usepackage[T1]{fontenc}
\usepackage{hyperref}
\usepackage{url}
\usepackage{booktabs}
\usepackage{amsfonts}
\usepackage{amsmath}
\usepackage{amssymb}
\usepackage{amsthm}
\usepackage{nicefrac}
\usepackage{microtype}
\usepackage{graphicx}
\usepackage{doi}

\newtheorem{theorem}{Theorem}
\newtheorem{definition}{Definition}
\newtheorem{axiom}{Axiom}
\newtheorem{corollary}{Corollary}
\newtheorem{proposition}{Proposition}

\title{Trustless Provenance Trees: A Game-Theoretic Framework for\\
Operator-Gated Blockchain Registries}

\author{
  \href{https://orcid.org/0009-0001-6031-4066}{\includegraphics[scale=0.06]{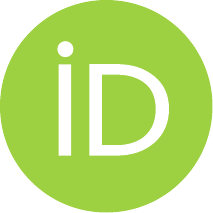}\hspace{1mm}Ian C. Moore, PhD}\\
  Founder, AnchorRegistry\\
  \texttt{imoore@anchorregistry.com}
}

\renewcommand{\shorttitle}{Trustless Provenance Trees}

\hypersetup{
pdftitle={Trustless Provenance Trees: A Game-Theoretic Framework for Operator-Gated Blockchain Registries},
pdfsubject={cs.GT},
pdfauthor={Ian C. Moore},
pdfkeywords={blockchain, provenance, game theory, mechanism design, cryptographic commitment, smart contracts, Ethereum},
}

\begin{document}
\maketitle

\begin{abstract}
We present a formal treatment of provenance trees — directed acyclic graphs of artifact registrations anchored immutably on a public blockchain — and introduce the operator trust problem: when a single privileged operator submits all on-chain registrations on behalf of users, the on-chain record alone cannot distinguish user-initiated registrations from unilateral operator actions. We resolve this through a dual-layer cryptographic commitment scheme in which two commitments derived from a single client-side secret key — binding the key to the tree root and to each unique registration identifier — make false attribution claims strictly dominated strategies. We prove correctness under standard cryptographic assumptions and establish honest behavior as the unique Nash equilibrium without relying on operator trust.
We further introduce and analyze the tree poisoning problem — adversarial attacks on users' provenance trees via fraudulent root registration, malicious child attachment, and tree identity spoofing. We characterize the closure properties of each attack variant and prove that a complete provenance tree integrity model requires three distinct mechanisms: cryptographic priority, governance cascade, and contract enforcement, each necessary and none individually sufficient.
The construction is deployed on Base (Ethereum L2) as AnchorRegistry, an immutable on-chain provenance registry. We provide gas complexity analysis demonstrating O(1) cost invariant to registry scale, and a trustless reconstruction algorithm recovering the complete registry from public event logs alone.
\end{abstract}


\section{Introduction}
\label{section:intro}

The proliferation of AI-generated content has created an urgent need for
infrastructure that can establish immutable provenance records for digital
artifacts. As AI systems train on, derive from, and recombine prior work,
the question of who created what, when, and from what prior artifacts becomes
both legally significant and technically difficult to answer. Blockchain
technology offers a natural substrate for such infrastructure: an append-only,
tamper-resistant ledger that can serve as a permanent timestamped record of
artifact existence and authorship.

A natural design for such a system is an \emph{operator-gated registry}: a
smart contract in which a privileged operator submits registrations on behalf
of users who have authenticated off-chain via a payment flow. This design
avoids requiring users to hold cryptocurrency or manage gas fees --- a
significant practical advantage --- but introduces a fundamental trust
question: how can a verifier, examining only the on-chain record, determine
whether a given registration was initiated by the purported user or by the
operator acting unilaterally?

We call this the \emph{operator trust problem}. Prior work on blockchain
timestamping~\cite{Hepp18, Gipp15} addresses proof of existence but does not
consider the attribution problem in operator-gated settings. NFT-based
ownership systems~\cite{Wang21} address ownership transfer but not artifact
provenance lineage. Decentralized identifier systems~\cite{W3C22} address
identity but not per-artifact initiation proof.

This paper makes the following contributions:

\begin{enumerate}
\item We formally define the \emph{provenance tree} as a mathematical object
      with seven structural properties (Section~\ref{section:tree}).

\item We formalize the operator trust problem as a strategic game and show
      that without additional mechanism design, the on-chain record is
      insufficient to resolve attribution disputes
      (Section~\ref{section:trust}).

\item We introduce the \emph{dual-layer commitment scheme} --- two commitments
      derived from a single client-side key that together constitute a
      trustless initiation proof --- and prove its correctness under standard
      cryptographic assumptions (Section~\ref{section:commitment}).

\item We analyze the resulting system as a mechanism design problem,
      demonstrate that false attribution claims are strictly dominated
      strategies, and introduce the \emph{tree poisoning game} analyzing
      adversarial attacks on other users' provenance trees
      (Section~\ref{section:equilibrium}).

\item We describe the implementation on Base (Ethereum L2) as AnchorRegistry,
      with a reference implementation provided as an open-source Python
      package~\cite{AnchorPy26} with full documentation~\cite{AnchorDocs26},
      including gas complexity analysis and trustless reconstruction
      (Section~\ref{section:implementation}).
\end{enumerate}

\section{The Provenance Tree}
\label{section:tree}

We begin by formalizing the provenance tree as a mathematical object.

\begin{definition}[Provenance Tree]
A \emph{provenance tree} $\mathcal{T} = (V, E, \lambda)$ is a directed acyclic
graph where:
\begin{itemize}
  \item $V$ is a finite set of artifact anchors (nodes)
  \item $E \subseteq V \times V$ is a set of directed provenance edges
  \item $\lambda: V \to \mathcal{M}$ assigns each node a metadata record from
        the metadata space $\mathcal{M}$
  \item There exists a unique root $r \in V$ with no incoming edges
  \item Every node $v \in V$ is reachable from $r$
\end{itemize}
\end{definition}

Each node $v \in V$ carries the following fields in its metadata record
$\lambda(v) \in \mathcal{M}$:

\begin{itemize}
  \item $\mathsf{arId}(v)$: a globally unique artifact identifier
  \item $\mathsf{type}(v)$: an artifact type from a fixed taxonomy
  \item $\mathsf{manifestHash}(v)$: SHA-256 hash of the artifact manifest
  \item $\mathsf{treeId}(v)$: a tree identity commitment (defined in Section~\ref{section:commitment})
  \item $\mathsf{tokenCommitment}(v)$: a per-anchor initiation commitment
        (defined in Section~\ref{section:commitment})
  \item $\mathsf{parentArId}(v)$: the $\mathsf{arId}$ of the parent node,
        empty for the root
\end{itemize}

We require the following structural properties of any valid provenance tree:

\begin{proposition}[Structural Properties of Provenance Trees]
\label{prop:structural}
A provenance tree $\mathcal{T}$ satisfies the following properties:
\begin{description}
  \item[P1 -- Immutability.] Once written to the blockchain, no node can be
        modified or deleted. The graph is append-only.

  \item[P2 -- Uniqueness.] $\forall u, v \in V, u \neq v \Rightarrow
        \mathsf{arId}(u) \neq \mathsf{arId}(v)$. Enforced by the
        \texttt{registered} mapping in the smart contract.

  \item[P3 -- Ancestry Integrity.] $\forall (u, v) \in E$: $u$ must be
        registered before $v$. Enforced by the \texttt{\_validateBase}
        function requiring \texttt{registered[parentArId]} prior to
        registration.

  \item[P4 -- Tree Membership.] $\forall v \in V$:
        $\mathsf{treeId}(v) = \mathsf{treeId}(r)$. All nodes share the
        same tree identity.

  \item[P5 -- Reconstructibility.] Given the contract address and deploy
        block number, the complete graph $\mathcal{T}$ is recoverable from
        the public blockchain event log in $O(|V|)$ time, with no dependency
        on off-chain infrastructure.

  \item[P6 -- Tree Ownership.] $\exists K$ such that
        $H(K \| \mathsf{arId}(r)) = \mathsf{treeId}(r)$, where $H$ is a
        cryptographic hash function and $K$ is the owner's secret key.

  \item[P7 -- Initiation Proof.] $\forall v \in V$ with
        $\mathsf{tokenCommitment}(v) \neq \mathbf{0}_{32}$:
        $\exists K$ such that $H(K \| \mathsf{arId}(v)) =
        \mathsf{tokenCommitment}(v)$.
\end{description}
\end{proposition}

Properties P1--P4 are enforced at the smart contract level and hold
unconditionally for any valid provenance tree. Properties P5--P7 are the
subject of the remainder of this paper: P5 in
Section~\ref{section:implementation}, and P6--P7 in
Section~\ref{section:commitment}.

\section{The Operator Trust Problem}
\label{section:trust}

\subsection{System Model}

In an operator-gated provenance registry, we have three classes of actors:

\begin{itemize}
  \item \textbf{Users} $\mathcal{U}$: individuals who wish to register
        artifacts. Users authenticate via an off-chain payment flow and
        receive an ownership token $K$ generated client-side.

  \item \textbf{Operator} $\mathcal{O}$: a privileged smart contract role
        that submits all on-chain registrations. The operator is the
        AnchorRegistry service itself. Only whitelisted operator wallets
        can call register functions.

  \item \textbf{Verifiers} $\mathcal{V}$: any party wishing to verify
        provenance claims, including other users, legal entities, AI
        systems, and compliance tools.

  \item \textbf{Adversaries} $\mathcal{A}$: external parties who may attempt
        to corrupt provenance claims, either against the operator or against
        other users' trees.
\end{itemize}

The operator submits every on-chain transaction on behalf of users.
Consequently, the \texttt{registrant} field in every \texttt{Anchored}
event always contains the operator wallet address --- not the user's address.
This is a deliberate design choice that removes the requirement for users to
hold cryptocurrency, but it introduces a fundamental attribution ambiguity.

\subsection{The Attribution Problem}

Without additional mechanism design, the on-chain record alone cannot
distinguish the following two scenarios:

\begin{enumerate}
  \item \textbf{Legitimate registration}: User $u$ authenticates, pays, and
        the operator submits the registration on $u$'s behalf.

  \item \textbf{Unilateral operator action}: The operator submits a
        registration without user authorization.
\end{enumerate}

In both scenarios, the on-chain \texttt{registrant} is the operator wallet.
The manifest hash, artifact type, and metadata are identical. A verifier
examining the blockchain record cannot distinguish these cases.

\subsection{The False Attribution Game}

We formalize this as a strategic game $\Gamma = (N, S, u)$:

\begin{itemize}
  \item \textbf{Players} $N = \{$User, Operator$\}$
  \item \textbf{User strategies} $S_U = \{$\emph{accuse}, \emph{silent}$\}$
  \item \textbf{Operator strategies} $S_O = \{$\emph{register-legitimate},
        \emph{register-unilateral}$\}$
\end{itemize}

Without the dual-layer commitment scheme (Section~\ref{section:commitment}),
the payoff matrix admits a problematic equilibrium: a user can falsely
accuse the operator of unilateral registration at no cost, since the
on-chain evidence is ambiguous. Similarly, an operator can register content
unilaterally without cryptographic detection.

The central insight of this paper is that by introducing appropriate
cryptographic commitments, we can transform this game such that false
accusation becomes a strictly dominated strategy.

\section{The Dual-Layer Commitment Scheme}
\label{section:commitment}

\subsection{Cryptographic Primitives}

We work with the following primitives:

\begin{axiom}[keccak256 Preimage Resistance]
\label{axiom:preimage}
For the hash function $H = \mathsf{keccak256}$: given $y = H(x)$,
finding $x$ is computationally infeasible. Formally, for any probabilistic
polynomial-time adversary $\mathcal{A}$:
\[
\Pr[H(\mathcal{A}(H(x))) = H(x)] \leq \mathsf{negl}(\lambda)
\]
where $\lambda$ is the security parameter.
\end{axiom}

\begin{axiom}[Client-Side Key Secrecy]
\label{axiom:secrecy}
The ownership token $K$ is generated client-side in the user's browser
and is never transmitted to the operator or any server. The operator
has no access to $K$.
\end{axiom}

\begin{axiom}[Registration Identifier Uniqueness]
\label{axiom:unique}
Each registration identifier $C_i$ (\textsf{arId}) is globally unique,
pre-generated via a reservation endpoint, and consumed atomically at
registration time.
\end{axiom}

\begin{axiom}[Operator Gate]
\label{axiom:gate}
Only whitelisted operator wallets can invoke registration functions.
The \texttt{onlyOperator} modifier enforces this at the contract level.
Outside parties have zero registration capability.
\end{axiom}

\begin{axiom}[Commitment Enforcement]
\label{axiom:enforcement}
The smart contract enforces:
\begin{itemize}
  \item Content anchors: $\mathsf{tokenCommitment} \neq \mathbf{0}_{32}$
        (enforced by \texttt{MissingTokenCommitment} revert)
  \item Governance anchors: $\mathsf{tokenCommitment} = \mathbf{0}_{32}$
        (hardcoded, not passable from outside)
\end{itemize}
\end{axiom}

\subsection{Construction}

Let $K \in \{0,1\}^{256}$ be a uniformly random ownership token generated
client-side. Let $R$ denote the root anchor identifier and $C_i$ denote
the $i$-th child anchor identifier. Define:

\begin{align}
T &= H(K \| R) \label{eq:treeid} \\
\Phi_i &= H(K \| C_i) \label{eq:commitment}
\end{align}

where $\|$ denotes concatenation, $T$ is the \emph{tree identity commitment}
written on-chain as \texttt{treeId}, and $\Phi_i$ is the \emph{per-anchor
initiation commitment} written on-chain as \texttt{tokenCommitment}.

The same key $K$ parameterizes both commitments. This is the \emph{unified}
property: a single credential suffices to prove both tree ownership (via $T$)
and per-anchor initiation (via all $\Phi_i$).

\subsection{Security Proofs}

We work under the cryptographic and system axioms stated in
Section~\ref{section:commitment}. All proofs are in the standard
cryptographic style: reductions to well-known assumptions with explicit
negligible functions.

\begin{theorem}[Tree Ownership]
\label{thm:ownership}
Let $T = H(K \| R)$ be the tree identity commitment stored on-chain,
where $K \in \{0,1\}^{256}$ is the secret ownership token generated
client-side, $R$ is the publicly known root anchor identifier, and $H$
denotes the keccak256 hash function.

Then, for any presented key $K'$:
\[
H(K' \| R) = T \implies \text{the presenter holds the original secret } K
\]
(with overwhelming probability, except with negligible probability in $\lambda$).
\end{theorem}

\begin{proof}
Assume, for contradiction, that there exists a PPT adversary $\mathcal{A}$
that, on input the public values $T$ and $R$, outputs $K'$ such that
$H(K' \| R) = T = H(K \| R)$ without knowledge of $K$.

Because $R$ is public metadata of the root anchor (the \texttt{arIdPlain}
field in the \texttt{Anchored} event), $\mathcal{A}$ knows both $T$ and $R$
and has produced a preimage $x' = K' \| R$ of $T$.

By Axiom~\ref{axiom:preimage}, for any PPT $\mathcal{A}$ and input
$x = K \| R$:
\[
\Pr\bigl[H\bigl(\mathcal{A}(H(x))\bigr) = H(x)\bigr] \leq \mathsf{negl}(\lambda).
\]
The adversary is precisely inverting $H$ on challenge $T$. Success
probability is negligible.

By Axiom~\ref{axiom:secrecy}, the operator has zero knowledge of $K$ at
any point. Only the legitimate holder of $K$ can satisfy the equation with
non-negligible probability. This establishes Property P6. \qed
\end{proof}

\begin{theorem}[Per-Anchor Initiation Proof]
\label{thm:initiation}
Let $\Phi_i = H(K \| C_i)$ be the per-anchor initiation commitment stored
on-chain for a content anchor, where $K \in \{0,1\}^{256}$ is the secret
ownership token generated client-side, $C_i$ is the globally unique artifact
identifier, and $H$ denotes the keccak256 hash function. By design,
$\Phi_i \neq \mathbf{0}_{32}$ for all content anchors.

Then, for any presented key $K'$:
\[
H(K' \| C_i) = \Phi_i \implies \text{the presenter initiated anchor } i
\]
(with overwhelming probability, except with negligible probability in $\lambda$).
\end{theorem}

\begin{proof}
Assume, for contradiction, that there exists a PPT adversary $\mathcal{A}$
that, on input the public values $\Phi_i$ and $C_i$, outputs $K'$ such that
$H(K' \| C_i) = \Phi_i = H(K \| C_i)$ without knowledge of $K$.

Because $C_i$ is public metadata (the \texttt{arIdPlain} field in the
\texttt{Anchored} event), $\mathcal{A}$ has produced a preimage
$x' = K' \| C_i$ of $\Phi_i$.

By Axiom~\ref{axiom:preimage}, for any PPT $\mathcal{A}$ and input
$x = K \| C_i$:
\[
\Pr\bigl[H\bigl(\mathcal{A}(H(x))\bigr) = H(x)\bigr] \leq \mathsf{negl}(\lambda).
\]

By Axiom~\ref{axiom:secrecy}, the operator has zero knowledge of $K$.
By Axiom~\ref{axiom:unique}, each $C_i$ is globally unique and
pre-generated, so the operator cannot reuse identifiers to fabricate
collisions. By Axiom~\ref{axiom:enforcement}, the contract rejects any
content registration with $\mathsf{tokenCommitment} = \mathbf{0}_{32}$,
so every content anchor carries a genuine non-zero $\Phi_i$.

Consequently, no PPT adversary can produce a valid $K'$ except with
negligible probability. Only the holder of $K$ can satisfy the equation.
This establishes Property P7. \qed
\end{proof}

\begin{theorem}[Governance Separation]
\label{thm:governance}
For any anchor $v \in V$ in a valid provenance tree:
\[
\mathsf{tokenCommitment}(v) = \mathbf{0}_{32}
\iff
v \text{ was registered by the operator in governance capacity.}
\]
\end{theorem}

\begin{proof}
($\Rightarrow$) By Axiom~\ref{axiom:enforcement}, the only functions that
write $\mathbf{0}_{32}$ into \texttt{tokenCommitment} are the governance
registration functions. These hardcode $\mathbf{0}_{32}$ internally; the
value is not a parameter and cannot be overridden by any external caller.
All governance functions are further gated by \texttt{onlyOperator}.

($\Leftarrow$) By Axiom~\ref{axiom:enforcement}, content registration
functions revert with \texttt{MissingTokenCommitment} if
$\mathsf{tokenCommitment} = \mathbf{0}_{32}$. No content anchor can carry
the zero value. An anchor carrying $\mathbf{0}_{32}$ must have been created
via a governance function.

The bidirectional implication holds strictly, enforced immutably by the
deployed bytecode on Base (Ethereum L2). \qed
\end{proof}

\begin{corollary}[False Accusation Impossibility]
\label{cor:false}
A user $u$ cannot successfully claim that the operator registered a
content anchor in $u$'s tree without $u$'s authorization.
\end{corollary}

\begin{proof}
Suppose $u$ makes such a claim regarding anchor $v$ with
$\mathsf{tokenCommitment}(v) = \Phi_i \neq \mathbf{0}_{32}$.

\textbf{Case 1}: $u$ produces $K$ such that $H(K \| C_i) = \Phi_i$.
By Theorem~\ref{thm:initiation}, $u$ holds $K$ and therefore initiated
the registration. The claim collapses.

\textbf{Case 2}: $u$ refuses to produce such $K$. This is equivalent to
asserting keccak256 preimage resistance is broken. Under
Axiom~\ref{axiom:preimage}, this claim is rejected as computationally
infeasible.

No third case exists. \qed
\end{proof}

\begin{corollary}[Unified Proof]
\label{cor:unified}
A single ownership token $K$ simultaneously satisfies both
Theorem~\ref{thm:ownership} and Theorem~\ref{thm:initiation} for all
anchors in the tree:
\[
H(K \| R) = T \quad \land \quad \forall i: H(K \| C_i) = \Phi_i.
\]
\end{corollary}

\begin{proof}
Immediate from Theorems~\ref{thm:ownership} and~\ref{thm:initiation},
observing that the same client-side secret $K$ parameterizes both
commitments (Equations~\ref{eq:treeid} and~\ref{eq:commitment}). \qed
\end{proof}

\section{Nash Equilibrium Analysis}
\label{section:equilibrium}

\subsection{The Transformed Game}

With the dual-layer commitment scheme in place, we revisit the false
attribution game $\Gamma$ from Section~\ref{section:trust}.

The commitment scheme introduces a \emph{verification challenge}: any
attribution dispute can be settled by demanding the production of $K$
such that $H(K \| C_i) = \Phi_i$. This transforms the payoff structure
across all four games analyzed in this section.

\subsection{The False Accusation Game}

Consider user $u$ contemplating a false accusation against the operator:

\begin{itemize}
  \item \textbf{Strategy: Accuse}. User claims operator registered anchor
        $v$ without authorization. The operator demands production of $K$.
        \begin{itemize}
          \item If $u$ produces $K$: verification passes, $u$ is shown to
                have initiated the registration. \emph{Payoff: self-incrimination.}
          \item If $u$ refuses: claim requires asserting keccak256 is broken.
                \emph{Payoff: claim dismissed.}
        \end{itemize}
  \item \textbf{Strategy: Silent}. User does not make false claims.
        \emph{Payoff: neutral.}
\end{itemize}

Since both outcomes of the \emph{Accuse} strategy are weakly dominated by
\emph{Silent}, false accusation is a strictly dominated strategy.

\subsection{The Enterprise Griefing Game}

Consider an enterprise user $u$ with an ACCOUNT anchor providing $n$
batch registrations, attempting to create $n$ identical registrations and
claim the operator acted unilaterally:

\begin{itemize}
  \item By Axiom~\ref{axiom:unique}, each registration $i$ receives a
        distinct $C_i$ via the reservation endpoint.
  \item Each anchor carries $\Phi_i = H(K_i \| C_i)$ for some $K_i$
        produced at registration time.
  \item To support the accusation, $u$ must produce $K_i$ for each anchor
        such that $H(K_i \| C_i) = \Phi_i$.
  \item Producing any $K_i$ proves $u$ initiated registration $i$.
\end{itemize}

The pattern of $n$ identical manifest hashes is also trivially detectable.
Griefing is therefore both self-incriminating and detectable, making it
a strictly dominated strategy regardless of batch size.

\subsection{The Landlord-Tenant Game}
\label{section:landlord}

We characterize the operator-user relationship as a mechanism design
problem with the following structure:

\begin{itemize}
  \item \textbf{Tenant} (user): has full agency over content lifecycle
        (register, retract). Both actions require non-zero $\Phi_i$,
        which only the user can produce.

  \item \textbf{Landlord} (operator): has governance authority over
        building integrity (REVIEW, VOID, AFFIRMED). All governance
        actions carry $\mathbf{0}_{32}$, hardcoded at contract level.

  \item \textbf{Lease} (smart contract): defines the exact authority
        of each party. Immutably enforced on Ethereum.
\end{itemize}

The Nash equilibrium of this game is the honest equilibrium: both parties
act within their defined roles, since any deviation is either impossible
(blocked by contract) or self-defeating (provably attributable via
commitment).

\begin{proposition}[Honest Equilibrium]
\label{prop:equilibrium}
Under the dual-layer commitment scheme, the unique Nash equilibrium of
the operator-user game is the honest equilibrium in which:
\begin{enumerate}
  \item Users initiate only registrations they intend to make
  \item The operator registers only user-authorized content and
        legitimate governance actions
\end{enumerate}
\end{proposition}

\begin{proof}
By Corollary~\ref{cor:false}, false accusation is strictly dominated.
By Axiom~\ref{axiom:secrecy} and Theorem~\ref{thm:initiation}, unilateral
operator registration of content anchors is detectable (the operator
cannot produce a valid $\Phi_i$ without $K$). Unilateral operator
governance actions (VOID, REVIEW, AFFIRMED) are legitimate by design and
carry the distinguishing sentinel $\mathbf{0}_{32}$. Therefore neither
party has a profitable deviation from the honest equilibrium. \qed
\end{proof}

\subsection{The Tree Poisoning Game}
\label{section:poisoning}

The preceding three games model adversarial behavior directed at
AnchorRegistry. We now analyze a qualitatively different attack class:
adversarial attacks on the provenance trees of legitimate users. This
is the \emph{tree poisoning problem}, and it is the primary threat to
the core user value proposition --- that a provenance tree constitutes
a tamper-evident IP lineage.

We identify three attack variants, each with distinct closure properties:

\subsubsection*{Variant 1: Fraudulent Root Registration}

An adversary $\mathcal{A}$ registers a fraudulent root anchor
$r_{\mathcal{A}}$ claiming to predate a legitimate root $r_u$ owned by
user $u$, then registers children under $r_{\mathcal{A}}$ to construct
a false IP priority claim.

\textbf{Closure mechanism: Cryptographic Priority.}

Each root produces a distinct treeId:
\[
T_u = H(K_u \| r_u) \neq H(K_{\mathcal{A}} \| r_{\mathcal{A}}) = T_{\mathcal{A}}
\]
since $K_u \neq K_{\mathcal{A}}$ and $r_u \neq r_{\mathcal{A}}$. The two
trees are cryptographically distinct and cannot be confused. Priority is
determined by block timestamp on the public blockchain --- the first
registered root is the legitimate root. This is an objective, tamper-proof
comparison requiring no trust in any party.

Furthermore, $\mathcal{A}$ cannot produce a tree carrying $T_u$ without
knowing $K_u$, by Theorem~\ref{thm:ownership}. The legitimate tree is
therefore unforgeable even if $\mathcal{A}$ knows $r_u$.

\textbf{Nash Equilibrium:} Fraudulent root registration is a strictly
dominated strategy. The adversary cannot claim $u$'s treeId without $K_u$,
and a competing fraudulent tree with a distinct treeId is immediately
distinguishable from the legitimate one.

\subsubsection*{Variant 2: Malicious Child Attachment}

The permissionless lineage model (Property P3) permits any registered
artifact to declare any prior artifact as its parent. An adversary
$\mathcal{A}$ may exploit this by registering a malicious child anchor
$v_{\mathcal{A}}$ with $\mathsf{parentArId}(v_{\mathcal{A}}) = \mathsf{arId}(v_u)$
for some legitimate anchor $v_u$ in user $u$'s tree.

\textbf{Important design note:} This variant is \emph{intentionally
permissionless}. The open attachment model enables organic provenance
growth --- a dataset registered by party A may legitimately be declared
as a parent by party B's model, without requiring A's consent, just as
academic citations require no consent from the cited author. Requiring
attachment consent would destroy this property.

\textbf{Closure mechanism: Governance Cascade.}

Malicious attachment is not prevented cryptographically but is remedied
operationally. The VOID governance action triggers a cascade that removes
the malicious subtree from public visibility:

\begin{enumerate}
  \item Any party may flag a malicious attachment to the operator.
  \item The operator registers a VOID anchor targeting the fraudulent root
        of the malicious subtree (not the legitimate parent).
  \item The VOID cascade suppresses all descendants of the voided anchor
        at the off-chain index layer.
  \item The on-chain record is never modified (P1 -- Immutability);
        suppression is exclusively off-chain.
\end{enumerate}

The key invariant is that VOID actions carry $\mathsf{tokenCommitment} =
\mathbf{0}_{32}$ (Theorem~\ref{thm:governance}), making them permanently
distinguishable from user-initiated content. The suppression action is
attributable to AR governance, not to the legitimate tree owner.

\textbf{Nash Equilibrium:} Malicious child attachment is a dominated
strategy. The attachment is detectable, attributable to $\mathcal{A}$
via $\Phi_{\mathcal{A}} = H(K_{\mathcal{A}} \| C_{\mathcal{A}})$, and
removable via VOID cascade without affecting the legitimate tree.

\subsubsection*{Variant 3: Tree Identity Spoofing}

An adversary $\mathcal{A}$ reads $T_u = \mathsf{treeId}(r_u)$ from the
on-chain event log and attempts to register new anchors carrying $T_u$
in order to appear as a legitimate member of user $u$'s tree.

\textbf{Closure mechanism: Contract Enforcement.}

By Axiom~\ref{axiom:gate}, only the operator wallet can invoke registration
functions. $\mathcal{A}$ has no direct contract access. Any attempt to
register an anchor carrying $T_u$ must pass through the operator, which
validates the tokenCommitment:
\[
\Phi = H(K \| C_i) \neq H(K_u \| C_i) \text{ unless } K = K_u.
\]
By Axiom~\ref{axiom:enforcement}, a zero tokenCommitment reverts. A
non-zero tokenCommitment computed without $K_u$ produces a value that
fails verification against $T_u$. Neither case admits a valid spoofed
anchor in $u$'s tree.

\textbf{Nash Equilibrium:} Tree identity spoofing is strictly dominated.
The contract physically prevents it without requiring any governance action.

\subsubsection*{Summary: Three Closure Mechanisms}

\begin{table}[h]
\centering
\begin{tabular}{llll}
\toprule
Attack Variant & Closure Mechanism & Cryptographic? & Governance? \\
\midrule
Fraudulent root & Cryptographic priority & Yes & No \\
Malicious child & VOID cascade & No & Yes \\
Tree ID spoofing & Contract enforcement & Yes & No \\
\bottomrule
\end{tabular}
\caption{Tree poisoning attack variants and their closure mechanisms.
A complete provenance tree integrity model requires all three mechanisms.
No single mechanism is sufficient to address all variants.}
\label{table:poisoning}
\end{table}

\begin{proposition}[Tree Integrity Completeness]
\label{prop:integrity}
The combination of cryptographic priority (Theorem~\ref{thm:ownership}),
governance cascade (VOID), and contract enforcement
(Axioms~\ref{axiom:gate} and~\ref{axiom:enforcement}) is necessary and
sufficient to close all three tree poisoning attack variants.
\end{proposition}

\begin{proof}
\textbf{Sufficiency} follows from the three subgame analyses above: each
variant is closed by its corresponding mechanism.

\textbf{Necessity}: Remove any single mechanism and a variant remains open.
Without cryptographic priority, fraudulent root registration is
indistinguishable from legitimate registration by block timestamp alone
(timestamps can be close; the treeId commitment provides the unique identity
signal). Without governance cascade, malicious child attachments persist
permanently on-chain. Without contract enforcement, an adversary with
operator cooperation could register spoofed tree members.

All three mechanisms are therefore necessary. \qed
\end{proof}

\section{Implementation}
\label{section:implementation}

\subsection{AnchorRegistry on Base L2}

The dual-layer commitment scheme is deployed as AnchorRegistry on Base
(Ethereum L2). The smart contract \texttt{AnchorRegistry.sol} implements
23 artifact types across 8 logical groups, three registration entry points
(\texttt{registerContent}, \texttt{registerGated}, \texttt{registerTargeted}),
and a 6-key access control architecture with a 7-day timelocked recovery
mechanism.

The \texttt{Anchored} event emitted on every registration contains:

\begin{verbatim}
event Anchored(
    string  indexed arId,
    address indexed registrant,
    ArtifactType    artifactType,
    string          arIdPlain,
    string          descriptor,
    string          title,
    string          author,
    string          manifestHash,
    string          parentArId,
    string  indexed treeId,
    string          treeIdPlain,
    bytes32         tokenCommitment
);
\end{verbatim}

The non-indexed \texttt{arIdPlain} and \texttt{treeIdPlain} fields enable
complete registry reconstruction from event logs without a database
dependency (Property P5). The indexed \texttt{treeId} topic enables
targeted tree-level queries via \texttt{eth\_getLogs}, supporting O(1)
lookup of any tree's complete history.

\subsection{Commitment Construction}

The client-side commitment construction proceeds as follows:

\begin{enumerate}
  \item Prior to payment, the client calls a \texttt{/reserve} endpoint
        to pre-generate a unique $C_i$, stored as PENDING.
  \item The client samples $\mathit{salt} \xleftarrow{\$} \{0,1\}^{256}$ and computes $K = \texttt{keccak256}(\mathit{salt})$.
  \item The client computes $\Phi_i = keccak256(K||Ci)$ using ethers.js.
  \item Only $\Phi_i$ and $C_i$ are transmitted; $K$ never leaves the browser.
  \item On payment confirmation, the webhook handler retrieves the PENDING
        $C_i$, passes $\Phi_i$ to \texttt{registerContent}, and atomically
        consumes the reservation.
\end{enumerate}

\subsection{Gas Complexity Analysis}

The commitment scheme introduces the following additional operations per
registration:

\begin{table}[h]
\centering
\begin{tabular}{lrr}
\toprule
Operation & Gas cost & Complexity \\
\midrule
\texttt{tokenCommitments[arId] = $\Phi_i$} & $\approx$20,000 & O(1) \\
\texttt{$\Phi_i$ == $\mathbf{0}_{32}$} check & $\approx$3 & O(1) \\
\texttt{bytes32} event emission & $\approx$375 & O(1) \\
\midrule
\textbf{Total added} & $\approx$20,378 & O(1) \\
\bottomrule
\end{tabular}
\caption{Additional gas cost per registration from the commitment scheme.
At Base L2 gas prices ($\approx$0.001 gwei), this represents a fraction
of a cent per registration. All operations are O(1) with respect to
registry size, tree depth, and anchor count.}
\label{table:gas}
\end{table}

\subsection{Trustless Reconstruction (Property P5)}

The complete provenance tree $\mathcal{T}$ is recoverable from public
blockchain data alone:

\begin{verbatim}
events = eth.get_logs({
    "address": CONTRACT_ADDRESS,
    "topics":  [ANCHORED_EVENT_SIGNATURE],
    "fromBlock": DEPLOY_BLOCK,
    "toBlock": "latest"
})
\end{verbatim}

Each event reconstructs one node. The \texttt{parentArId} field
reconstructs edge structure. The indexed \texttt{treeId} topic enables
targeted per-tree reconstruction. This algorithm runs in O($|V|$) time
and requires no off-chain infrastructure.

The \texttt{anchorregistry} Python package~\cite{AnchorPy26} implements
this reconstruction and provides the \texttt{authenticate\_tree} function
for trustless verification, with full documentation at~\cite{AnchorDocs26}:

\begin{verbatim}
result = authenticate_tree(
    ownership_token = K,
    root_ar_id      = R
)
# result["authenticated"] == True iff:
# H(K || R) == treeId  AND
# H(K || C_i) == tokenCommitment_i for all i
\end{verbatim}

\section{Related Work}
\label{section:related}

\textbf{Blockchain timestamping.} OriginStamp~\cite{Hepp18} and
Bernstein~\cite{Gipp15} provide proof of existence via Merkle tree
aggregation on Bitcoin. These systems batch multiple hashes into a single
transaction, trading immediacy for cost efficiency. Neither system
addresses the operator trust problem, the tree poisoning problem, or
per-record initiation proof. The use of Bitcoin mainnet makes
individual-record timestamping economically infeasible; the emergence of
L2 networks, surveyed comprehensively in~\cite{Sgu21}, fundamentally
alters this tradeoff by enabling per-transaction settlement at fractions
of a cent.

\textbf{NFT ownership systems.} Non-fungible tokens~\cite{Wang21} provide
proof of asset ownership via wallet signatures. However, NFTs address
transferable ownership of tokenized assets rather than the provenance
lineage of digital artifacts. The speculative secondary market properties
of NFTs introduce regulatory concerns absent from a pure provenance registry.
NFT ownership requires the user to hold cryptocurrency and manage a wallet,
eliminating the operator-gated UX advantage.

\textbf{Decentralized identifiers.} The W3C DID specification~\cite{W3C22}
provides decentralized identity infrastructure. DIDs address
\emph{identity} but not artifact \emph{provenance}: they do not model
parent-child relationships between artifacts, do not provide per-artifact
initiation proof, and do not address the operator trust problem or tree
poisoning in gated submission contexts.

\textbf{Mechanism design for blockchain protocols.} The design of incentive-
compatible blockchain protocols has received significant attention~\cite{Rou20}.
EIP-1559~\cite{But19} represents a canonical application of mechanism design
to transaction fee markets. Our work applies similar game-theoretic reasoning
to the provenance attribution and tree integrity problems, complementing
this literature.

\textbf{Commitment schemes.} Cryptographic commitment schemes are
well-studied~\cite{Ped91}. Our construction applies standard keccak256
commitments in a novel dual-layer architecture that binds tree-level
ownership and per-anchor initiation through a single client-side secret,
enabling a unified proof system for provenance attribution and tree
integrity.

\section{Conclusion}
\label{section:conclusion}

We have presented a formal framework for provenance trees as mathematical
objects, introduced the dual-layer commitment scheme to resolve the
operator trust problem, and characterized the tree poisoning problem with
a complete analysis of its three attack variants and their closure
mechanisms.

The central results are: (1) false attribution claims against the operator
are strictly dominated strategies under the dual-layer commitment scheme;
(2) fraudulent root registration is closed by cryptographic priority; (3)
malicious child attachment is closed by governance cascade; (4) tree
identity spoofing is closed by contract enforcement; and (5) all three
closure mechanisms are necessary --- no single mechanism is sufficient.

The honest equilibrium --- in which users initiate only their own
registrations, the operator registers only authorized content and legitimate
governance actions, and adversaries have no profitable tree poisoning
strategy --- is the unique Nash equilibrium of the resulting system.

The construction is O(1) in gas cost with respect to registry scale,
enabling deployment on Ethereum L2 at fractions of a cent per registration.
The complete registry is reconstructible from public event logs in O($|V|$)
time with no off-chain infrastructure dependency.

The provenance tree model, dual-layer commitment scheme, and tree integrity
completeness result represent novel contributions to the design of
trustless attribution infrastructure for the AI era, where the question
of who created what, from what prior work, and when, is increasingly both
legally significant and technically difficult to answer without
cryptographic enforcement.

\section*{Acknowledgments}

The author used AI writing assistance (Anthropic Claude) in drafting this manuscript. All theoretical contributions, formal proofs, system design decisions, and architectural choices are the author's own. The reference implementation described in this paper is deployed and operational at https://anchorregistry.com.


\begin{thebibliography}{1}

\bibitem{Hepp18}
T. Hepp, A. Schoenhals, C. Gondek, and B. Gipp,
\textit{OriginStamp: A Blockchain-Backed System for Decentralized Trusted Timestamping},
it -- Information Technology, 60(5-6):273--281, 2018.

\bibitem{Gipp15}
B. Gipp, N. Meuschke, and A. Gernandt,
\textit{Decentralized Trusted Timestamping using the Crypto Currency Bitcoin},
Proceedings of the iConference, 2015.

\bibitem{Wang21}
Q. Wang, R. Li, Q. Wang, and S. Chen,
\textit{Non-Fungible Token (NFT): Overview, Evaluation, Opportunities and Challenges},
arXiv:2105.07447, 2021.

\bibitem{W3C22}
M. Sporny, D. Longley, M. Sabadello, D. Reed, O. Steele, and C. Allen,
\textit{Decentralized Identifiers (DIDs) v1.0},
W3C Recommendation, July 2022.

\bibitem{Rou20}
T. Roughgarden,
\textit{Transaction Fee Mechanism Design for the Ethereum Blockchain: An Economic Analysis of EIP-1559},
arXiv:2012.00854, 2020.

\bibitem{But19}
V. Buterin,
\textit{Blockchain Resource Pricing},
Ethereum Research, April 2019.

\bibitem{Sgu21}
C. Sguanci, R. Spatafora, and A. M. Vergani,
\textit{Layer 2 Blockchain Scaling: A Survey},
arXiv:2107.10881, 2021.

\bibitem{Ped91}
T. P. Pedersen,
\textit{Non-Interactive and Information-Theoretic Secure Verifiable Secret Sharing},
CRYPTO 1991, LNCS 576, pp. 129--140, Springer, 1991.

\bibitem{AnchorPy26}
I. C. Moore,
\textit{anchorregistry: Trustless Python Client for the AnchorRegistry Provenance Chain},
PyPI, v0.1.3, 2026.
Available: \url{https://pypi.org/project/anchorregistry/}

\bibitem{AnchorDocs26}
I. C. Moore,
\textit{anchorregistry Documentation},
ReadTheDocs, 2026.
Available: \url{https://anchorregistry.readthedocs.io/en/latest/}

\end{thebibliography}
\end{document}